\documentclass[11pt]{article}
\usepackage{fullpage}
\let\labelindent\relax
\usepackage{enumitem}
\usepackage{times}
\usepackage{color}
\usepackage{amsmath}
\usepackage{amsthm}
\usepackage{amsfonts}
\usepackage{xspace}
\usepackage{caption}
\usepackage{cite}
\usepackage{subfigure}
\usepackage{hyperref}
\usepackage{graphicx}
\usepackage{authblk}

\def \cD     {{\cal D}}
\def \cO     {{\cal O}}
\def \cP     {{\cal P}}

\def \cX     {{\cal X}}

\newcommand{\EE}{\mathbb{E}}
\newcommand{\ie}{\textit{i.e.,}\xspace}
\newcommand{\ed}{\stackrel{\mathrm{def}}{=}}

\newcommand{\LZA}{\textsc{LZA}}
\newcommand{\DIFFENC}{\textsc{Difference-Encode}}
\newcommand{\edge}{\delta}

\newcommand{\ER}{Erd\H{o}s-R\'enyi }

\newtheorem{Theorem}{Theorem}
\newtheorem{Lemma}[Theorem]{Lemma}

\newcommand{\temp}{T}
\newcommand{\ignore}[1]{}

\newcommand{\itmadjo}{\setlength{\itemindent}{-3mm}}

\title{Automata and Graph Compression}
\author[1,2]{Mehryar Mohri\thanks{mohri@cims.nyu.edu}}
\author[2]{Michael Riley\thanks{riley@google.com}}
\author[3]{Ananda Theertha Suresh\thanks{asuresh@ucsd.edu}}
\affil[1]{Courant Institute of Mathematical Sciences}
\affil[2]{Google Research}
\affil[3]{University of California, San Diego}

\ignore{
\author{
\authorblockN{Mehryar Mohri}
\authorblockA{Courant Institute and Google Research\\
{mohri@cs.nyu.edu}}
\and
\authorblockN{Michael Riley}
\authorblockA{Google Research\\
{riley@google.com}}
\and
\authorblockN{Ananda Theertha Suresh}
\authorblockA{UCSD\\
{asuresh@ucsd.edu}}
}
}

\begin{document}
\maketitle

\begin{abstract}
  We present a theoretical framework for the compression of automata, 
  which are widely used in speech processing and other natural language 
  processing tasks. The framework extends to graph compression.
\ignore{
  Every day, untold gigabytes of data are generated. How to
  efficiently store this massive amount of data has received wide
  attention with most of the work focusing on compressing sequences.
  For example, Lempel-Ziv coding, Huffman coding, and arithmetic
  coding are sequence compression algorithms.  Much of the generated
  data, however, has structure associated with it.  We focus on two
  particular structures, \emph{directed graphs} and \emph{finite
    automata} with applications that include webgraphs and speech
  processing, respectively.}
  Similar to stationary ergodic processes, we formulate a
  probabilistic process of graph and automata generation that captures real world phenomena and provide a universal compression scheme
  \LZA\ for this probabilistic model.  Further, we show that \LZA\
  significantly outperforms other compression
  techniques such as \emph{gzip} and the UNIX \emph{compress} command for several
  synthetic and real data sets.
\ignore{
theertha:Sometimes it beats bzip2 and sometimes
it does not, should we mention it here?
}
\end{abstract}
\noindent

\section{Introduction}
\label{sec:introduction}
The rapid generation of data by search engines and popular online sites, which has been reported to be in the order of hundreds of petabytes, requires efficient storage  mechanisms and better compression algorithms. 
Similarly, sophisticated models on mobile devices for tasks such as speech-to-text conversion often need large storage space 
and memory constraints on these devices demand better compression algorithms. Furthermore, downloading these models requires
high bandwidth so transmitting a compressed version saves communication costs. Hence there is a need for efficient data 
compression both at the data warehouse level (petabytes of data) and at a device level (megabytes of data).
 
Most of the current compression techniques have been
developed for sequential data.  For example, Huffman coding and
arithmetic coding are optimal compression schemes when the underlying
sequence is distributed independently (i.i.d.) according to some known
distribution~\cite{CoverT06}.  If the sequence is not generated
according to an i.i.d. process, but generated from a stationary
ergodic process, then \emph{Lempel-Ziv} schemes are asymptotically
optimal~\cite{ZivL77,HanselPS92}.  In practice, a combination of these schemes
are often used.  For example, the UNIX \emph{compress} command implements
Lempel-Ziv-Walsh (LZW) and \emph{gzip} combines Lempel-Ziv-77 (LZ77)
and Huffman coding.

However, data is often structured.  For example, in a webgraph, a node
represents a URL and a directed edge between two nodes indicates that
one URL has a link to another.  In social networks, a node represents
a user and an edge between two nodes indicates that they are friends.
Finite automata and transducers are widely used in speech recognition,
and a variety of other language processing tasks such as machine translation, information extraction, and parsing~\cite{MohriPR02}.
For example, in speech-processing automata, a path may correspond to a possible sentence in a language model or in a set of recognizer hypotheses (a so-called \emph{lattice}). Often these data sets are very large. For web graphs, there are tens of billions of web pages to choose from.  For speech processing, a large-alphabet language model may have billions of word edges. Hence structured data compression is useful in practice.

A natural question is to ask how one can exploit the structure in data to develop better compression algorithms? Can one do better
than serializing the data and applying algorithms for sequence compression?
Surprisingly, these questions and the compression of structured data
have received little attention. Motivated by previous examples, we focus on automata compression and, as a corollary,
graph compression.

 \cite{ApostolicoD09, GrabowskiB10,
  AnhM10} studied webgraph compression empirically.
 Theoretical webgraph compression
was first studied by \cite{AdlerM01} who proposed a scheme that uses
a \emph{minimum spanning tree} to find similar nodes to compress.
However, they showed that many generalizations of their problem are NP hard.
Motivated by probabilistic models,~\cite{ChoiS09,ChoiS12} showed that
arithmetic coding can be used to near-optimally compress (the
\emph{structure} of) graphs generated by the \ER model.

Automata compression empirically has been studied by~\cite{Daciuk00,DaciukP06,DaciukW11}. However, we are not aware of any theoretical work
focused on automata compression. Our goal is three-fold: $(i)$ 
propose a probabilistic model for automata that captures
real world phenomena, $(ii)$ provide a provable universal
compression algorithm, and $(iii)$ show experimentally that the
algorithm fares well compared to techniques such as gzip and
compress. 
We note that our probabilistic model can be viewed as a generalization of \ER graphs~\cite{AlonS92}.

The rest of the paper is organized as follows: in
Section~\ref{sec:automata}, we describe automata and their properties.
In Section~\ref{sec:model}, we describe our probabilistic model and
show how it captures many real-world applications. In
Section~\ref{sec:algorithm}, we describe our proposed algorithm \LZA,
prove its optimality and in
Section~\ref{sec:experiments}, we demonstrate the algorithm's
practicality in terms of its degree of compression.\ignore{ and running time.}

\section{Directed Graphs and Finite Automata}
\label{sec:automata}

A directed graph is a pair $(Q,\edge)$
where $Q =\{1,2,3,\ldots, n\}$ is the set of nodes and 
$\edge\colon Q \rightarrow Q^*$ 
is the set of edges where for every node $q$,
$\edge(q)$ is the set of nodes to which it is connected.
Note that our notation for directed graphs is chosen
to harmonize with finite automata.

Automata generalize graphs.  An unweighted automaton $A$ is a
$5$-tuple $(Q,\Sigma,\edge,q_i,F)$ where $Q =\{1,2,\ldots ,n\}$ is the
set of states, $\Sigma = \{1,2,\ldots,m\}$ is a finite alphabet,
$\edge\colon Q \times \Sigma \rightarrow Q^*$ is the transition
function, $q_i \in Q$ is the initial state, and $F \subseteq Q$ are
the final states.  The transitions from state $q$ by label $a$ to
states $\{q'_1,q'_2,\ldots\}$ are given by $\edge(q,a)
=\{q'_1,q'_2,\ldots\}$.  If there is no transition by label $a$, then
$\edge(q,a)=\emptyset$.  We use $E \subset Q \times \Sigma \times Q$ 
to denote the set of all transitions $(q,a,q')$ and $E[q]$ to denote the set of all transitions from state $q$.
\ignore{
\begin{figure}[t]
\centering
\includegraphics[scale=0.6]{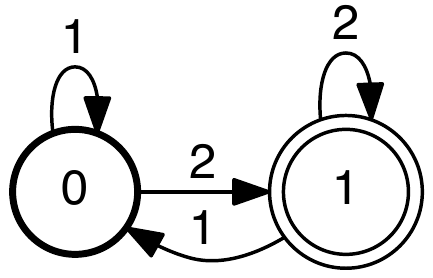}
{\footnotesize{\caption{An example automaton.}}}
\label{fig:fsa}
\end{figure}
\begin{figure}[t]
\centering
\includegraphics[scale=0.5]{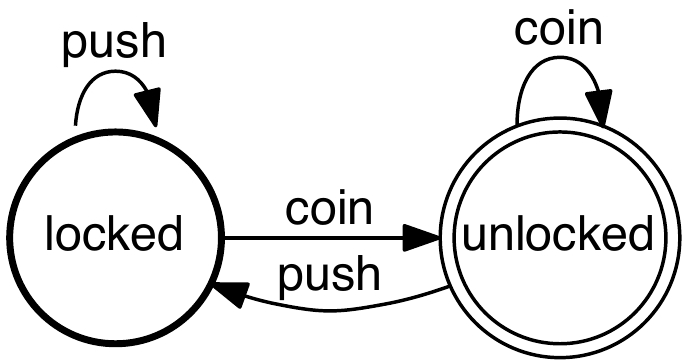}
{\footnotesize{\caption{Subway turnstile automaton.}}}
\label{fig:lock}
\end{figure}
}

\begin{figure}[t]
\centering
\begin{minipage}{.5\textwidth}
  \centering
\includegraphics[scale=0.6]{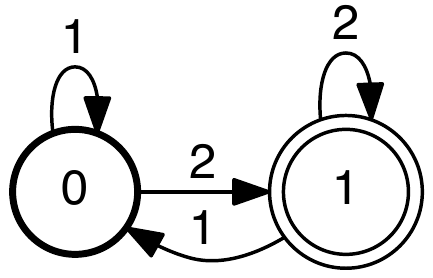}
\caption{{{An example automaton.}}}
\label{fig:fsa}
\end{minipage}%
\begin{minipage}{.5\textwidth}
  \centering
\includegraphics[scale=0.5]{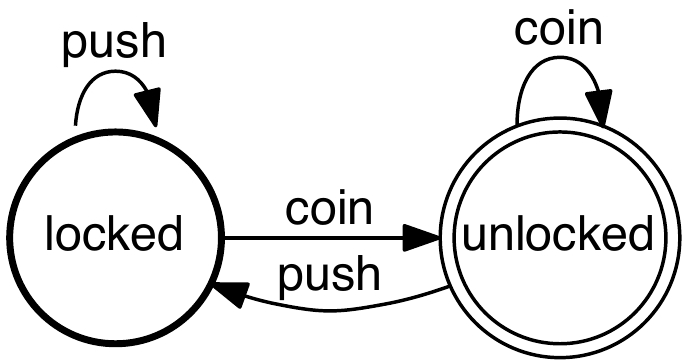}
\caption{{{A subway turnstile automaton.}}}
\label{fig:lock}
\end{minipage}
\end{figure}
\begin{figure}[t]
\begin{center}
\includegraphics[scale = 0.6]{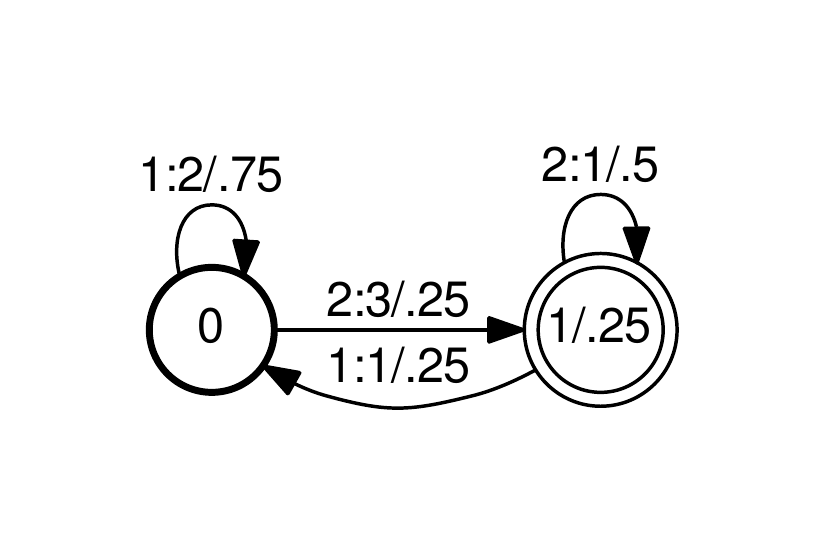}
\end{center} 
\caption{An example weighted transducer.}
\label{fig:fst}
\end{figure}
\begin{figure}[t]
\centering     
\subfigure{\label{fig:a}\includegraphics[width=32mm,angle=-90]{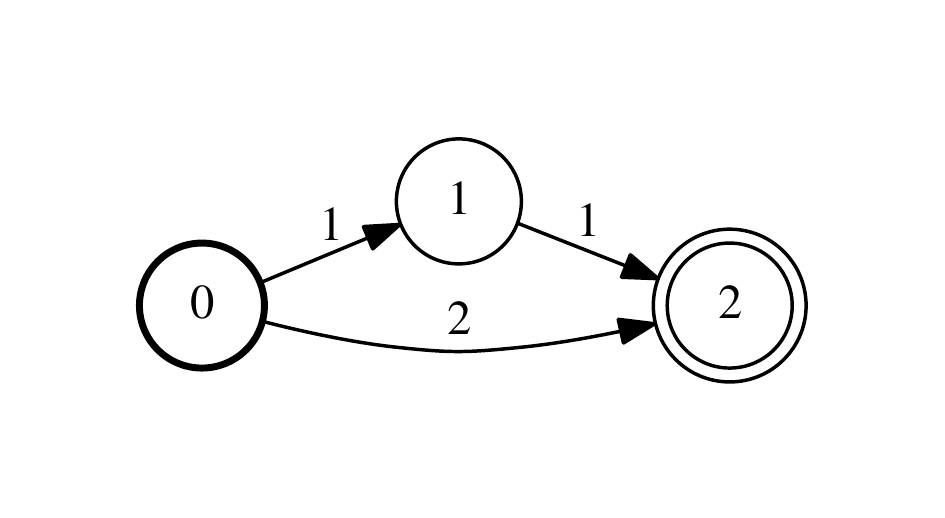}}
\subfigure{\label{fig:b}\includegraphics[width=32mm,angle=-90]{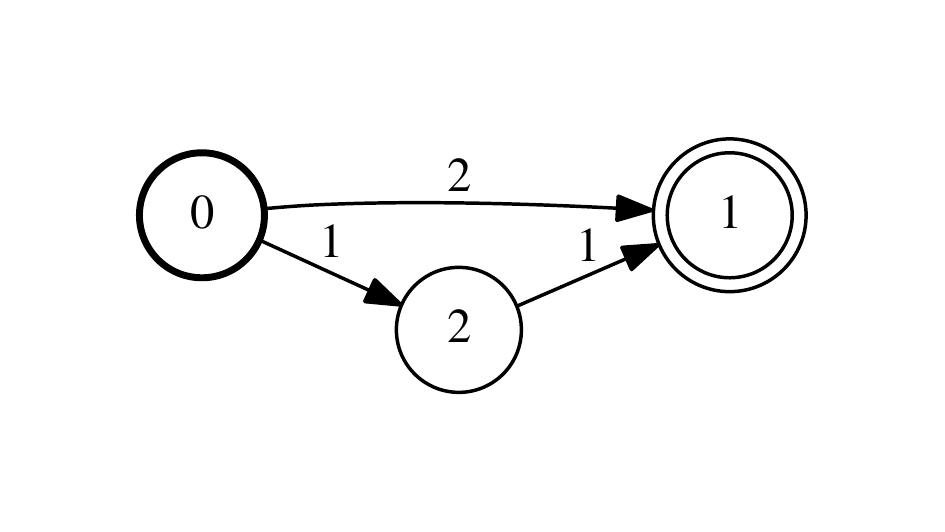}}
\caption{An example of isomorphic automata. The above two automata are same under the permutation $0\to0, 1\to2$, and $2\to 1$}
\label{fig:iso}
\end{figure}

An example of an automaton is given in Figure \ref{fig:fsa}.
State 0 in this simple example is the initial state (depicted with the bold circle)
and state 1 is the final state (depicted with double circle).
The strings \emph{12} and \emph{222} are among those accepted by this automaton.
By using symbolic labels on this automaton in place of the usual integers, 
as depicted in Figure~\ref{fig:lock},
we can interpret this automaton as the operation of a subway turnstile.
It has two states \emph{locked} and \emph{unlocked} and 
actions (alphabet) \emph{coin} and \emph{push}.
If the turnstile is in the locked state and you push, 
it remains locked and if you insert a coin, 
it becomes unlocked. 
If it is unlocked and you insert a coin it remains 
unlocked, but if you push once it becomes locked.

Note that directed graphs form a subset of 
automata with $\Sigma = \{1\}$ and hence we focus on automata compression. 
Furthermore, to be consistent with the existing automata literature, 
we use states to refer to nodes and transitions to refer to edges in both graphs
and automata going forward.

A main motivation to study automata is their application in speech and natural language processing.
In some circumstances, transitions may be generalized to have an
output label and a weight as well as the usual input label. 
Such automata, called \emph{weighted finite state transducers} (FSTs),
are extensively used in these fields~\cite{RocheS95, MohriPR02, AlabauCVJ07}. An example of an FST is given in Figure \ref{fig:fst}.
The string \emph{12} is among those accepted by this transducer. For this input, the transducer outputs
the string \emph{23} and has weight $.046875$ (transitions weights $0.75$ times $0.25$ times final weight $0.25$).

We propose an algorithm for unweighted automata compression. For FSTs,
we use the same algorithm by treating the input-output label pair as a
single label. If the automaton is weighted, we just add the weights at
the end of the compressed file by using some standard representation.

\section{Random automata compression}
\label{sec:model}

\subsection{Probabilistic model}
Our goal is to propose a probabilistic model for automata generation that captures real world phenomena. To this end, we first review probabilistic models on sequences and draw connections to probabilistic models for automata.
\subsubsection{Probabilistic processes on sequences}
We now define i.i.d. sampling of sequences.
Let $x^n_1$ denote an $n$-length sequence $x_1,x_2\ldots x_n$.
If $x^n_1$ are $n$ independent samples from a distribution $p$ over $\cX$,
then $p(x^n_1) = \prod^n_{i=1} p(x_i)$.
Note that under i.i.d. sampling, the index of the sample has no importance, \ie 
\begin{equation*}
p(X_i = x) = p(X_j = x), \, \forall 1\leq i,j \leq n, x \in \cX.
\end{equation*}
stationary ergodic processes generalizes i.i.d. sampling.
For a stationary ergodic process $p$ over sequences
\begin{equation*}
p(X^m_i = x^m_i) = p(X^{m+j}_{i+j} = x^m_i), \forall i,j,m, x^{m}_i.
\end{equation*}
Informally stationary ergodic processes are those for which only the relative position of the indices matter and not the actual ones.

\subsubsection{Probabilistic processes on automata}
\label{sec:prob_automata}
Before deriving models for automata generation, we first discuss 
an invariance property of automata that is useful in practice. 
The set of strings accepted by an automaton and the time and
space of its use are not affected by the state numbering. 
Two automata are \emph{isomorphic} if they coincide modulo a
renumbering of the states. Thus, automata $(Q, \Sigma, \edge, q_i, F)$
and $(Q', \Sigma, \edge', q'_i, F')$ are isomorphic, if there is a
one-to-one mapping $f\colon Q \to Q'$ such that $f(\edge(q, a)) =
\edge'(f(q), a)$, for all $q \in Q$ and $a \in \Sigma$, $f(q_i) =
q'_i$, and $f(F) = F'$, where $f(F) = \{f(q)\colon q \in F\}$.

Under stationary ergodic processes, two sequences with the same order of
observed symbols have the same probabilities. Similarly we wish to
construct a probabilistic model of automata such that any two
isomorphic automata have the same probabilities, since 
the state numbering does not have explicit importance, 
For example, the
probabilities of automata in Figure~\ref{fig:iso} are the same.

There are several probabilistic models of automata and 
graphs that satisfy this property. Perhaps the most studied 
random model is the \ER model $G(n,p)$, 
where each state is connected to every other state independently with probability $p$~\cite{AlonS92}. 
Note that if two automata are isomorphic then the \
\ER model assigns them the same probability. 
The \ER model is analogous to i.i.d. sampling on sequences. 
We wish to generalize the \ER model to more realistic models of automata. 

Since the state numbering is to be disregarded,
the only possible dependence of transitions from a state would be by the paths leading to that state. 
This arises naturally in language modeling tasks. 
\ignore{
We generalize the \ER model to a Markov-equivalent sampling model for
automata.  The need for such models arises naturally in several
applications. 
}
For example in an $n$-gram model,
a state might have an outgoing transition with label \emph{Francisco} or \emph{Diego}
only if it has a input transition with label \emph{San}.  
This is an example where we have restrictions on 
paths of length $2$. In general, we may have restrictions on paths of any length $\ell$. 

We define an $\ell$-\emph{memory} model for automata as follows.  Let
$h^{\ell}_q$ be the set of paths of length at most $\ell$
leading to the state $q$. The probability distribution of transitions
from a state depends on the paths leading to it.  Let $\edge(q,*) \ed
\edge(q,1), \edge(q,2),\ldots, \edge(q,m)$.
\begin{align*}
p(A)
& = p(\edge(1,*), \edge(2,*), \ldots, \edge(n,*)) 
\propto \prod^n_{q=1} p(\edge(q,*) | h^\ell_q).
\end{align*}
Similarly, transitions leaving a state $q$ dissociate into marginals
conditioned on the history $h^{\ell}_q$ and probability that $q' \in \edge(q,a)$ also dissociates into marginals.
 \begin{align*}
  p(\edge(q,*) | h^\ell_q) 
& = \prod_{a \in \Sigma} p(\edge(q,a) |h^\ell_q) \\
& = \prod_{a \in \Sigma} \prod^n_{q'=1} p(\mathbb{I}(q' \in \edge(q,a)) |h^\ell_q),
 \end{align*}
where $\mathbb{I}(q' \in \edge(q,a))$ is the indicator of the event
$q' \in \edge(q,a)$.
Note that the probabilities are defined with proportionality. 
This is due to the probabilities possibly not adding to one. 
Thus we have a constant $Z$ to ensure that it is a probability distribution.
\begin{align*}
p(A) &= p(\edge(1,*), \edge(2,*), \ldots, \edge(n,*))\\
& =  \frac{1}{Z}\prod^n_{q=1} p(\edge(q,*) | h^\ell_q) \\
& = \frac{1}{Z} \prod^n_{q=1} 
\prod_{a \in \Sigma} p(\edge(q,a) | h^\ell_q).
\end{align*}
Note that $\ell$-\emph{memory} models assign the same probability to
automata that are isomorphic. In our calculations, we restrict $\ell$ to make the model tractable.

Note that sequences form a subset of automata as follows. 
For a sequence $x^n$ over alphabet $\Sigma$, consider the automata representation with states $Q = \{1,2,\ldots n\}$, initial state $q_i = 1$, final state $F = \{n\}$, alphabet $\Sigma$, and transition function $\delta(i,x_i) = i+1$ and $\delta(i,x)= \phi$ for all $x \neq x_i$. Informally, every sequence can be represented as an automaton with line as the underlying structure. Furthermore, note that 
the probability that two isomorphic automata should have the same probability is same as stating the indices in sequences do not have explicit meaning (stationary ergodic property).
\ignore{
Furthermore, the probabilistic model is
similar in flavor of graphical models.}

\subsection{Entropy and coding schemes}

A compression scheme is a mapping from $\mathcal{X}$ to $\{0,1\}^*$
such that the resulting code is prefix-free and can be uniquely
recovered. For a coding scheme $c$, let $l_c(x)$ denote the length of
the code for $x \in \cX$.  It is well-known that the expected number
of bits used by any coding scheme is the entropy of the distribution,
defined as $H(p) \ed \sum_{x\in \cX} p(x) \log \frac {1}{p(x)}$. The
well known Huffman coding scheme achieves this entropy up-to one
additional bit. For $n$-length sequences arithmetic coding is used,
%
which achieves compression up-to entropy with few additional bits of
error. 
\ignore{We conclude this section by noting that for a product
distribution $p_1\cdot p_2\cdot \ldots \cdot p_s$ over $\cX_1 \times
\cX_2 \times \ldots \times \cX_s$, the entropy is
 \begin{equation*}
H(p_1 \cdot p_2 \cdot \ldots \cdot p_s) =  \sum^s_{i=1} H(p_i).
\end{equation*}}

The above-mentioned coding methods such as Huffman coding and
arithmetic coding require the knowledge of the underlying
distribution. In many practical scenarios, the underlying distribution
may be unknown and only the broader class to which the distribution
belongs may be known. For example, we might know that the given
$n$-length sequence is generated by i.i.d.\ sampling of some unknown
distribution $p$ over $\{1, 2, \ldots, k\}$. The objective of a
universal compression scheme is to asymptotically achieve $H(p)$ bits
per symbol even if the distribution is unknown. A coding scheme $c$
for sequences over a class of distributions $\cP$ is called universal
if
\begin{equation*}
\limsup_{n \to \infty} \max_{p \in \cP}\frac{ \mathbb{E}[l_c(X^n)] - H(X^n)}{n} = 0.
\end{equation*}
The normalization factor in the above definition is $n$, as the number
of sequences of length $n$ increases linearly with $n$. For automata
and graphs with $n$ states (denoted by $A_n$) we choose a scaling scaling factor of $n^2$ as the number
of automata scales as $\exp(n^2)$. We call a coding scheme $c$ for
automata over a class of distributions $\cP$ universal if
\begin{equation*}
 \limsup_{n \to \infty} \max_{p \in \cP} \frac{\EE[l_c(A_n)] - H(A_n)}{n^2} = 0 .
\end{equation*}
We now describe the algorithm \LZA. Note that the algorithm does not
require the knowledge of the underlying parameters or the
probabilistic model.

\section{Algorithm for automaton compression}
\label{sec:algorithm}

Our algorithm recursively finds substructures over states and uses a
Lempel-Ziv subroutine. Our coding method is based on two auxiliary
techniques to improve the compression rate: Elias-delta coding and
coding the differences. We briefly discuss these techniques and
their properties before describing our algorithm.

\subsection{Elias-delta coding and coding the differences}

Elias-delta coding is a universal compression scheme for integers~\cite{Elias75}. 
To represent a positive integer $x$, Elias-delta codes use
$ \lfloor \log x \rfloor + 2 \lfloor  \log \lfloor \log x \rfloor +1 \rfloor +1$
bits.
\ignore{
The encoding and decoding algorithms for Elias-delta codes are as follows.
\begin{center}
\fbox{\begin{minipage}{1.0\textwidth}
Algorithm \textsc{Elias-encode} \newline
\textbf{Input:} Integer $x$.
\begin{enumerate}
\item 
Separate $x$ into the highest power of $2$ it contains $(2^{N'})$ and the remaining $N'$ binary digits.
\item 
Write $N = N' + 1$ in binary and subtract $1$ from the number of bits required to write $N$ in binary and prepend that many zeros.
\item
Append the remaining $N'$ binary digits to this representation of $N$.
 \end{enumerate}
\end{minipage}}
\end{center} 
\begin{center}
\fbox{\begin{minipage}{1.0\textwidth}
Algorithm \textsc{Elias-decode} \newline
\textbf{Input:} Integer $x$.
\begin{enumerate}
 \item 
Read and count zeros from the stream until you reach the first one. Call this count of zeros $L$.
\item 
Considering the one that was reached to be the first digit of an integer, with a value of $2L$, read the remaining $L$ digits of the integer. Call this integer $N$.
\item
Put a one in the first place of our final output, representing the value $2N-1$. Read and append the following $N-1$ digits.
 \end{enumerate}
\end{minipage}}
\end{center} 
}
To obtain a code over $\mathbb{N} \cup \{0\}$, we replace $x$ by $x+1$and use Elias-delta codes.

We now use Elias-delta codes to obtain to code sets of integers. Let $x_1,x_2,\ldots,x_m$ be integers such that $0 \leq x_1 \leq x_2
\leq \cdots \leq x_m \leq n$.  We use the following algorithm to code
$x_1,x_2,\ldots,x_m$. The decoding algorithm follows from
\textsc{Elias-decode}~\cite{Elias75}.
\begin{center}
\fbox{\begin{minipage}{1.00\textwidth}
Algorithm \DIFFENC \newline
\textbf{Input:} Integers $0 \leq x_1 \leq x_2
\leq \cdots \leq x_m \leq n$.
\begin{enumerate}
\itmadjo
 \item 
Use \textsc{Elias-encode} to code $x_1-0$, $x_2-x_1$,  \ldots $x_d-x_{d-1}$.
 \end{enumerate}
\end{minipage}}
\end{center} 
\begin{Lemma}[Appendix~\ref{app:different}]
\label{lem:difference}
For integers such that $0\leq x_1 \leq x_2, \ldots x_d \leq n$,
\DIFFENC\ uses at most
\begin{equation*}
d \log \frac{n+d}{d} + 2d \log\Bigl( \log \frac{n+d}{d} +1 \Bigr) + d 
\end{equation*}
bits.
\end{Lemma}
We first give an example to illustrate \DIFFENC's usefulness.
Consider graph representation using
adjacency lists.  For every source state, the order in which the
destination states are stored does not matter. For example, if
state $1$ is connected to states $2$, $4$, and $3$, it suffices to
represent the unordered set $\{2, 3, 4 \}$. In general if a state is
connected to $d$ out of $n$ states, then it suffices to encode the
ordered set of states $y_1, y_2, \ldots, y_d$ where $1 \leq y_1 \leq
y_2\leq y_3 \ldots y_d \leq n$. The number of such possible sets is
$\binom{n}{d}$.  If the state-sets are all equally likely, then the
entropy of state-sets is $ \log {n \choose d} \approx d \log
\frac{n}{d}$.

If each state is represented using $\log n$ bits, then $d \log n > d
\log \frac{n}{d}$ bits are necessary, which is not optimal.  However,
by Lemma~\ref{lem:difference}, \DIFFENC\ uses $d \log
\frac{n + d}{d} (1 + o(1)) \approx d \log \frac{n}{d}$, and hence is
asymptotically optimal.  Furthermore, the bounds in
Lemma~\ref{lem:difference} are for the worst-case scenario and in
practice \DIFFENC\ yields much higher savings.  A
similar scenario arises in \LZA\ as discussed later.

\subsection{LZA}

We now have at our disposal the tools needed to design a variant of
the Lempel-Ziv algorithm for compressing automata, which we denote by
\LZA.  Let $d_q \ed |E[q]|$  be the number of transitions from state
$q$ and let transitions in $E[q] = \{(q,a_1,q_1),(q,a_2,q_2),\ldots, (q,a_{d_q},q_{d_q})\}$ are ordered as follows:
for all $i$, $q_i \leq q_{i+1}$ and if $q_i = q_{i+1}$ then $a_i < a_{i+1}$.
\ignore{$(a_1, \edge(q, a_1)), (a_2, \edge(q, a_2)), \ldots, (a_{d_q},
\edge(q, a_{d_q}))$ be the list of non-empty labeled transitions
leaving $q$ ordered as follows: For all $i$, $\edge(q, a_i) \leq
\edge(q, a_{i + 1})$, and if $\edge(q, a_i) = \edge(q, a_{i + 1})$,
then $a_i < a_{i + 1}$.}

The algorithm is based on the observation that the ordering of the
transitions leaving a state does not affect the definition of an
automaton and works as follows. The states of the automaton are
visited in a BFS order. For each state visited, the set of outgoing
transitions are sorted based on their destination state. Next, the
algorithm recursively finds the largest overlap of the sets of
transitions that match some dictionary element and encodes the pair
(matched dictionary element number, next transition), and adds the
dictionary element to $\temp_d$, alphabet of the transition to $\temp_{\Sigma}$, and  the transition to $\temp_\edge$. It also updates the dictionary element by adding a new
dictionary element (matched dictionary element number, next
transition) to the dictionary.  Finally it encodes $\temp_d$,
$\temp_\edge$ using \DIFFENC\ and encodes each element in $T_{\Sigma}$ using
$\lceil \log m \rceil$ bits.

\vskip .1in
\noindent
\fbox{\begin{minipage}{1.00\textwidth} Algorithm
    \LZA\ \newline
    \textbf{Input:} The transition label function $\edge$ of the automaton. \\
    \textbf{Output:} Encoded sequence $S$.
\begin{enumerate}[leftmargin=*]
 \item 
Set dictionary $D = \emptyset$.
\item 
Visit all states $q$ in BFS order. For every state $q$ do:
\begin{enumerate}
\item 
Code $d_q$ using $\lceil \log nm \rceil$ bits.
\item 
Set $\temp_d = \emptyset$, $\temp_{\Sigma} = \emptyset$, and $\temp_\edge = \emptyset$. 
\item 
Start with $j=1$ in $E[q] = \{(q,a_1,q_1),(q,a_2,q_2),\ldots, (q,a_{d_q},q_{d_q})\}$ and continue till $j$ reaches $d_q$.
\begin{enumerate}
\item 
Find largest $l$ such that $(a_j,q_j),
\ldots,(a_{j + l},q_{j + l}) \in D$. Let this dictionary element be $d_r$.
\item
Add  $(a_j,q_j)
,\ldots
,(a_{j + l + 1},q_{j + l + 1})$ to $D$.
\item
Add $d_r$ to $\temp_d$, $q_{j + l + 1}$ to $\temp_\edge$, and $a_{j + l + 1}$ 
to $T_{\Sigma}$.
\end{enumerate}
\item
Use \DIFFENC\ to encode $\temp_d$, $\temp_\edge$ and encode each element in $T_{\Sigma}$ using $\lceil \log m\rceil $ bits. Append these sequences to $S$
\end{enumerate}
\item Discard the dictionary and output $S$.
 \end{enumerate}
\end{minipage}}
\vskip .1in 

We note that simply compressing the unordered sets $\temp_d$ and
$\temp_\edge$ suffices for unique reconstruction and thus \DIFFENC\ is
the natural choice.  Observe that \DIFFENC\ is a succinct
representation of the dictionary and does not affect the way
Lempel-Ziv dictionary is built.  Thus the decoding algorithm follows immediately from retracing the steps in \textsc{LZA} and LZ78 decoding algorithm.

If \DIFFENC\ is not used, the number
of bits used would be approximately $|D| \log |D| + |D| \log n $,
which is strictly greater than that number of bits in
Lemma~\ref{lem:LZA}. Furthermore, we did consider several other
natural variants of this algorithm where we difference-encode the
states first and then serialize the data using a standard Lempel-Ziv
algorithm. However, we could not prove asymptotic optimality for those
variants. Proving their non-optimality requires constructing
distributions for which the algorithm is non-optimal and is not the
focus of this paper.

We first bound the number of bits used by \LZA\ in terms of the size
of the dictionary $|D|$, the number of states $n$, and the alphabet
size $m$. This bound is independent of the underlying
probabilistic model. Next, we proceed to derive probabilistic bounds.

\begin{Lemma}[Appendix~\ref{app:LZA}]
\label{lem:LZA}
The total number of bits used by \LZA\ is at most
\begin{multline*}
|D| \left[ \log (n + 1)  +  \log \left( \nu + 1 \right)  +  2 \log (\log (n + 1)  +  1) \right] 
+ |D| \left[ 2   \log  \left(  \log \left( \nu + 1
    \right)  +  1  \right)  +  2  +   \lceil \log m \rceil   \right]
+  n \lceil \log nm \rceil,
\end{multline*}
where $\nu = \frac{n^2}{|D|} $.
\end{Lemma}

\subsection{Proof of optimality}
\label{sec:optim}

In this section, we prove that \LZA\ is asymptotically optimal for the
random automata model introduced in Section~\ref{sec:model}.
Lemma~\ref{lem:LZA} gives an upper bound on the number of bits used in
terms of the size of the dictionary $|D|$. 
We now present a lower bound on the entropy in terms of $D$ which will help us prove this
result.  The proof is given in Appendix~\ref{app:lower}.
\begin{Lemma}
\label{lem:lower}
\LZA\ satisfies
 \begin{equation*}
H(p) \geq  \EE[|D|] \left[  \log \frac{\EE[|D|]}{n}  - m^{\ell} - \log \left( \frac{n^2m}{\EE[|D|]}  +  1 \right) -1\right].
\end{equation*}
\end{Lemma}
The above result together with Lemma~\ref{lem:LZA} implies 
\begin{Theorem}[Appendix~\ref{app:main}]
\label{thm:main}
 If $2^{m^{\ell}} = o\left(\frac{\log n}{\log \log n} \right)$, 
then \LZA\ is a universal compression algorithm.
\end{Theorem}

\section{Experiments}
\label{sec:experiments}

\subsection{Automaton structure compression}

\LZA\ compresses automata, but for most applications,
it is sufficient to compress the automata structure.
We convert \LZA\ into \textsc{LZA}$_S$, an algorithm for automata 
structure compression as follows. 
We first perform a breadth first search (BFS) with the initial state as
the root state
and relabel the states in their BFS visitation order.
We then run \LZA\ with the following modification.
In step $2$, for every state $q$ we divide the transitions from $q$ into 
two groups, $T^q_{\text{old}}$ transitions whose destination
states have been traversed before in \LZA\ and $T^q_{\text{new}}$,
transitions whose destinations have not been traversed. 
Note that since the state numbers are ordered based on a BFS visit, 
the destination state numbers in $T^q_\text{new}$ are $1,2,\ldots n$, and
can be recovered easily while decoding and thus need not be stored. 
Hence, we run step $2b$ in \LZA\ only on transitions in 
$T^q_{\text{old}}$. For $T^q_{\text{new}}$, we just compress the transition labels 
using LZ78. 

Since each destination state can appear in $T^q_{\text{new}}$ only once, 
the number of transitions in $\cup_{q} T^q_\text{new} \leq n$,
Since this number is $\ll n^2$, the normalization factor in the definition 
of universal compression algorithm for automata, 
the proof of Theorem~\ref{thm:main} extends to \LZA$_s$.
Since for most applications, it is sufficient to compress 
to the automata structure, we implemented \textsc{LZA}$_s$ in C++ and
added it to the \emph{OpenFst} open-source library~\cite{AllauzenRSSM07}.

\subsection{Comparison}
The best known convergence rates of all Lempel-Ziv 
algorithms
for sequences are $\cO \left( \frac{\log \log n}{\log n} \right)$
and \LZA\ has the same convergence rate under the $\ell$-memory probabilistic model. 

However in practice data sets have finitely many states and the underlying automata
may not be generated from an $\ell$-memory probabilistic model. To prove the 
practicality of the algorithm, we compare \textsc{LZA$_{s}$} with the Unix compress 
command (LZ78) and gzip (Lempel-Ziv-Walsh and Huffman coding)
for various synthetic 
and real data sets. 

\subsubsection{Synthetic Data}
While the $\ell$-memory probabilistic model illustrates 
a broad class of probabilistic models on which
\LZA\,is universal, generating samples from an $\ell$-memory 
model is difficult as the normalization factor $Z$ is hard to compute. We therefore 
test our algorithm on a few simpler synthetic data sets. In all our experiments the 
number of states is $1000$ and the results are averaged over $1000$ runs. 

Table~\ref{tab:syn} summarizes our results for a few synthetic data sets, specified in bytes.
Note that one of the main advantages of \LZA$_S$ over existing algorithms is that \LZA$_S$ just compresses the structure, which is sufficient for applications
in speech processing and language modeling. Furthermore, note that to obtain the actual automaton from the structure we need 
the original state numbering, 
which can be specified in $n \log n$ bits, which is less than $1250$ bytes in our experiments.
Even if we add $1250$ bits to our results in Table~\ref{tab:syn}, \LZA\
still performs better than gzip and compress.

We run the algorithm on four different synthetic data sets $G_1,G_2,A_1,A_2$.
$G_1$ and $A_1$ are models with a uniform out-degree distribution 
over the states and $G_2$ and $A_2$ are models with a non-uniform out-degree distribution:
\newline
\indent $G_1$: directed \ER graphs where we randomly generate transitions between
every source-destination pair with probability $1/100$. \newline
\indent $A_1$: automata version of \ER graphs, where there is a transition between every two states 
with probability $1/100$ and the transition labels are chosen independently 
from an alphabet of size $10$ for each transition. 
\newline
\indent $G_2$: We first assign each state a class $c \in\{1,2,\ldots 1000\}$
randomly. We connect every two states $s$ and $d$ with probability $1/(c_s+c_d)$.
This ensures that the graph has degrees varying from $2$ to $\log 1000$.
\newline
\indent $A_2$: we generate the transitions as above and we label each transition to be a deterministic
 function of the destination state. This is similar to $n$-gram models, where the 
the destination state determines the label. Here again we chose $|\Sigma| = 10$.
\noindent
\begin{table}[t]
\begin{center}
\begin{tabular}{| c  | c | c | c | c |  c|}
  \hline                       
  Class &   \textsc{LZA}$_S$ & compress & gzip  & \textsc{LZA}$+$gzip \\
 $G_1$ & $18260$ & $22681$ & $23752$ &  $17320$ \\
 $A_1$ & $21745$ & $33478$ & $31682$  & $21108$ \\ 
 $G_2$ & $2536$ & $4994$ & $4564$ & $2443$ \\ 
  $A_2$ & $3027$ & $6707$ & $5546$ &  $2940$ \\ 
  \hline  
\end{tabular}
\end{center}
\caption{Synthetic data compression examples (in bytes).}
\label{tab:syn}
\end{table}
\ignore{
\noindent
\begin{tabular}{| c  | c | c | c | c | c | c|}
  \hline                       
  Class &   \textsc{LZA}$_S$ & compress & gzip  & bzip2 &\textsc{LZA},gzip \\
 $G_1$ & $18260$ & $22681$ & $23752$ &  $17320$ \\
 $A_1$ & $21745$ & $33478$ & $31682$  & $19288$  & $21108$ \\ 
 $G_2$ & $2536$ & $4994$ & $4564$ & $3023$ & $2443$ \\ 
  $A_2$ & $3027$ & $6707$ & $5546$ & $3662$ & $4190$ \\ 
  \hline  
\end{tabular}
\vspace{2ex}}

Note that \LZA\ always performs better than the standard 
Lempel-Ziv-based algorithms gzip and compress.
Note that algorithms designed with specific knowledge of the underlying model can achieve 
better performance. For example, for $G_1$,  arithmetic coding can be used to obtain 
a compressed file size of $n^2 h(0.01)/8 \approx 10000$ bytes. 
However the same algorithm would not perform well for $G_2$ or $A_2$.

\subsubsection{Real-World Data}
We also tested our compression algorithm on a variety of `real-world'
automata drawn from various speech and natural language
applications. These include large speech recognition language models
and decoder graphs~\cite{MohriPR02}, text normalization grammars for text-to-speech~\cite{TaiSS11},
speech recognition and machine translation lattices~\cite{IglesiasABGR11}, and pair $n$-gram
grapheme-to-phoneme models~\cite{BisaniN08}. We selected approximately eighty such
automata from these tasks and removed their weights and output labels
(if any), since we focus here on unweighted automata. Figure~\ref{fig:real_cmp}
shows the compressed sizes of these automata, ordered by their uncompressed 
(adjacency-list) size rank,
with the same set of compression algorithms
presented in the synthetic case. At the smallest sizes, gzip out-performs \textsc{LZA}, 
but after about 100 kbytes in compressed size, \textsc{LZA} is better. Overall, 
the combination of \textsc{LZA} and gzip performs best.

\begin{figure}[t]
\centering
\includegraphics[scale=0.4]{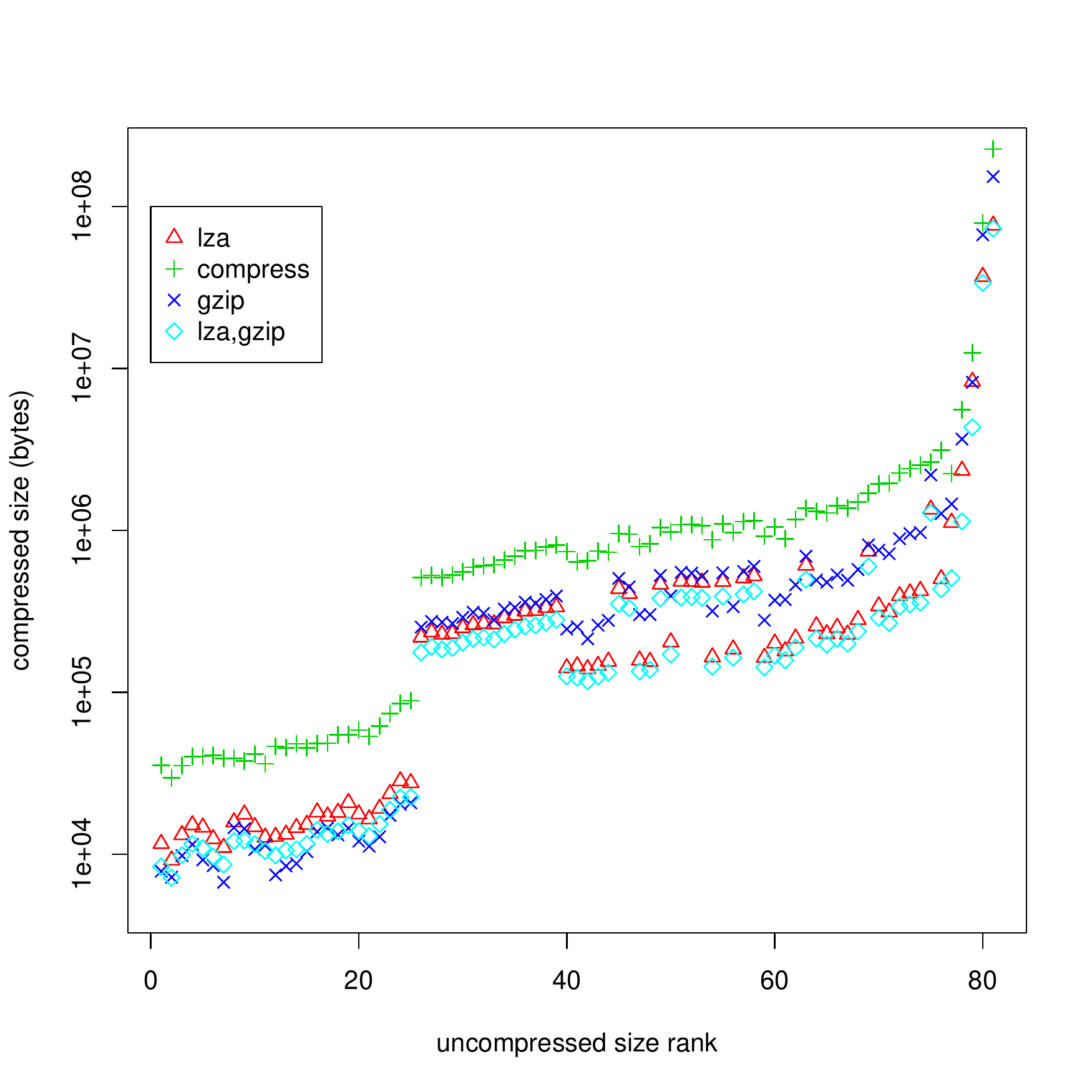}
\caption{Real-world compression examples.}
\label{fig:real_cmp}
\end{figure}

\ignore{
\section{Conclusion and Future Work}
\label{sec:conclusion}
Motivated by practical applications in speech processing
and language modeling we studied automata compression. 
We proposed an $\ell$-memory probabilistic model for
automata and graphs that captures many of the real world scenarios
and proposed a universal compression algorithm that has the same convergence
rate as that of Lempel-Ziv algorithms for sequences. 
Further by experiments on synthetic and real data we showed that the algorithm performs well.
We wish to develop on the fly algorithms that can uncompress parts of the automata
and . Faster algorithms?
}
\section{Acknowledgements}
\label{sec:acknowledgements}
We thank Jayadev Acharya and Alon Orlitsky for helpful discussions.
\bibliographystyle{plain}
\bibliography{abr,masterref}
\newpage
\appendix

\section{Proof of Lemma~\ref{lem:difference}}
\label{app:different}
  Since $0$ is included in the set, the number of bits used to
  represent $x$ is upper bounded by $\theta(x) = \log (x + 1) + 2 \log
  (\log (x + 1) + 1) \rfloor + 1$. Observe that $\theta$ is a concave
  function since both $\log$ and $x \mapsto \log (\log x)$ are
  concave.  Let $x_0 = 0$. Then, by the concavity of $\theta$, the
  total number of bits $B$ used can be bounded as follows:
\begin{align*}
B 
& \leq \sum_{i = 1}^d \theta(x_i - x_{i - 1})\\
& = d \, \sum_{i = 1}^d \frac{1}{d} \theta(x_i - x_{i - 1})\\
& \leq d \, \theta\Big(\frac{1}{d} \sum_{i = 1}^d x_i - x_{i - 1} \Big)
= d \, \theta\Big(\frac{1}{d} x_n \Big)
\leq d \, \theta\Big(\frac{n}{d} \Big),
\end{align*}
where we used for the last inequality $x_n \leq n$ and the fact that
$\theta$ is an increasing function. This completes the proof
of the lemma.

\section{Proof of Lemma~\ref{lem:LZA}}
\label{app:LZA}
  Let $k_q$ be the number of elements added to the dictionary when
  state $q$ is visited by \LZA.  The maximum value of the destination
  state is $n$. Thus, by Lemma~\ref{lem:difference}, the number of
  bits used to code $\temp_\edge$ is at most
\begin{equation*}
  \sum_{q = 1}^n \Big( k_q \theta\Big(\frac{n}{k_q} \Big) +  k_q \lceil \log m \rceil \Big),
\end{equation*}
where $\theta(\cdot)$ is the function introduced in the proof of
Lemma~\ref{lem:difference}.  Similarly, since the maximum value of any
dictionary element is $|D|$, by Lemma~\ref{lem:difference}, the number
of bits used to code $\temp_d$ is at most
\begin{equation*}
  \sum_{q = 1}^n  k_q \theta\bigg(\frac{|D|}{k_q} \bigg).
\end{equation*}
By concavity these summations are maximized when $k_q =
\frac{|D|}{n}$ for all $q$. Plugging in that expression in
the sums above yields the following upper bound on the
maximum number of bits used:
\begin{equation*}
  |D| \theta(\nu) +  |D| \lceil \log m \rceil + |D| \theta(n). 
\end{equation*}
Additionally, this number must be augmented by $n \lceil \log nm
\rceil$ since $\lceil \log nm \rceil$ bits are used to encode each
$d_q$, which completes the proof.

\section{Proof of Lemma~\ref{lem:lower}}
\label{app:lower}
One of the main technical tools we use is Ziv's inequality,
which is stated below.
\begin{Lemma}[Variation of Ziv's inequality]
\label{lem:ziv}
For a probability distribution $p$ over non-negative integers with mean $\mu$,
\begin{equation*}
 H(p) \leq \log(\mu + 1)  +  1.
 \end{equation*}
\end{Lemma}
The next lemma bounds the probability of disjoint events under different distributions.
\begin{Lemma}
\label{lem:disjoint}
If $A_1, A_2, \ldots, A_k$ be a set of disjoint events. Then for a set of
distributions $p_1,p_2,\ldots p_r$,
\begin{equation*}
  \sum_{i = 1}^k \sum_{j = 1}^r p_j(A_k) \leq  \sum^r_{j=1} 1 = r.
\end{equation*}
\end{Lemma}
We now lower bound $H(p)$ in terms of the number of dictionary
elements.

 Let $d_q \ed |E[q]|$  be the number of transitions from state
$q$ and let transitions in \newline $E[q] = \{(q,a_1,q_1),(q,a_2,q_2),\ldots, (q,a_{d_q},q_{d_q})\}$ are ordered as follows:
for all $i$, $q_i \leq q_{i+1}$ and if $q_i = q_{i+1}$ then $a_i < a_{i+1}$.
\ignore{
  Let $d_q$ be the number of transitions leaving $q$ with non-empty
  destination states and $a_1, a_2, \ldots, a_{d_q}$ the labels of
  these transitions once sorted in a monotonically non-decreasing
  order of their destination states.} To simplify the discussion, we
  will use the shorthand $e_{q,i} \ed (q,a_i,q_i)$. Then, by the
  definition of our probabilistic model,
\begin{align*}
  \log p(A) 
 & = \sum_{q = 1}^n\log  p(e_{q,1}, e_{q,2}, \ldots, e_{q,d_q} | h^{\ell}_q) - \log Z .
\end{align*}
We group the transitions the way \LZA\ constructed the dictionary.
Let $\{D_{q,i}\}$ be the set of dictionary elements added when state
$q$ is visited during the execution of the algorithm.  For a
dictionary element $D_{q,i}$, let $s_{q,q'}$ be the starting $e_{q,i}$
and $t_{q,q'}$ the terminal $e_{q,i}$.  Then, by the independence of
the transition labels and the fact that $Z \geq 1$,
\begin{align*}
   \log p(A)  \leq  \sum_{q = 1}^n \sum_{D_{q,i}} \log  p(e_{q,s_{q,i}},e_{q,s_{q,i} + 1},\ldots e_{q,t_{q,i}} | h^{\ell}_q).
\end{align*}
Let $g_{q,i} = t_{q,i} - s_{q,i}$.  We group them now with $g_{q,i}$
and $s_{q,i}$. Let $\cD(s, g)$ be the set of dictionary elements
$D_{q,i}$ with $s_{q,i} = s $ and $g_{q,i} = g$ and let $c_{s, g}$
be the cardinality of that set: $c_{s, g} = | \cD(s, g) |$. Then,
by Jensen's inequality, we can write
\begin{align*}
  &  \sum_{q = 1}^n \sum_{D_{q,i}} \log  p(e_{q,s_{q,i}},e_{q,s_{q,i} + 1},\ldots e_{q,t_{q,i}} | h^{\ell}_q) \\
  & = \sum_{q = 1}^n \sum_{s = 1}^n \sum_{g = 1}^n \sum_{D_{q,i} \in \cD(s, g)} \log  p(e_{q,s},e_{q,s + 1},\ldots e_{q,s + g} | h^{\ell}_q) \\
  & = \sum^n_{s=1} \sum_{g=1}^n c_{s, g} \frac{1}{c_{s,g}}\sum_{q = 1}^n   \sum_{D_{q,i} \in \cD(s, g)} \log  p(e_{q,s},e_{q,s + 1},\ldots e_{q,s + g} | h^{\ell}_q)  \\
  & \leq \sum_{s = 1}^n \sum_{g = 1}^nc_{s,g}  \log \frac{1}{c_{s,g}}\sum_{q = 1}^n   \sum_{D_{q,i} \in \cD(s, g)}   p(e_{q,s},e_{q,s + 1},\ldots e_{q,s + g} | h^{\ell}_q)  \\
  & \leq \sum_{s = 1}^n \sum_{g = 1}^n c_{s,g} \log
  \frac{2^{m^{\ell}}}{c_{s,g}}.
\end{align*}
where the last inequality follows by Lemma~\ref{lem:disjoint}, the
fact that the events in each summation are disjoint and mutually
exclusive and that the number of possible histories $h^{\ell}_q$ is
$\leq 2^{m^{l}}$.  We now have $\sum_{s,g} c_{s,g}= |D|$.  Thus,
\begin{align*}
  & \sum_{s = 1}^n \sum_{g = 1}^n c_{s,g} \log \frac{2^{m^{\ell}}}{c_{s,g}}\\
  & = |D| m^{\ell}  +  \sum_{s = 1}^n \sum_{g = 1}^n c_{s,g} \log \frac{1}{c_{s,g}} \\
  & = |D| m^{\ell} - |D| \log |D|  +  |D| \sum_{s = 1}^n \sum_{g = 1}^n \frac{ c_{s,g}}{D} \log \frac{|D|}{c_{s,g}} \\
  & = |D| m^{\ell} - |D| \log |D| + |D| H(c_{s,g}).
\end{align*}
Let $c_s$ and $c_g$ be the projections of $c_{s,g}$ into first and second coordinates.
Then, we can write
\begin{align*}
 H(c_{s,g}) \leq H(c_s)  +  H(c_g)
 & \leq \log n  +  H(c_g).
\end{align*}
Using $\sum_{s,g} c_{s,g} g \leq n^2 m$, by Ziv's inequality, the
following holds: $H(c_g) \leq \log \big( \frac{n^2m}{|D|} + 1 \big)$.
Combining this with the previous inequalities gives
\begin{equation*}
\log  p(A) \leq |D| \left[ m^{\ell}  +  \log \left( \frac{n^2m}{|D|}  +  1 \right)  +  1 - \log \frac{|D|}{n} \right].
\end{equation*}
Taking the expectation of both sides, next using
the concavity of $|D| \mapsto - |D| \log (|D|)$ and Jensen's inequality yield
\begin{equation*}
 H(p) \geq  \EE[|D|] \left[  \log \frac{\EE[|D|]}{n}  - m^{\ell} -
   \log \left( \frac{n^2m}{\EE[|D|]}  +  1 \right) -1 \right].
\end{equation*}

\section{Proof of Theorem~\ref{thm:main}}
\label{app:main}
We first upper bound $\EE[|D|]$ using Lemma~\ref{lem:lower}.
\begin{Lemma}
\label{lem:dictbound}
 For the dictionary $D$ generated by \LZA
 \begin{equation*}
\EE[|D|] \leq \frac{10n^{2}m\log(m + 1)2^{m^{\ell}}}{\log n}.
 \end{equation*}
\end{Lemma}
\begin{proof}
  An automaton is a random variable over $n^{2}$ transition labels
  each taking at most $m^m + 1$ values, hence $H(p) \leq n^{2} m \log
  (m + 1)$. Combining this inequality with Lemma~\ref{lem:lower} yields
\begin{equation*}
 n^{2} m \log(m + 1) \geq  
\EE[|D|] \left[  \log \frac{\EE[|D|]}{n}  - {m^{\ell}} - \log \left(
    \frac{n^2m}{\EE[|D|]}  +  1 \right) -1 \right].
\end{equation*}
Now, let $U = \frac{10n^{2}m\log(m + 1)2^{m^{\ell}}}{\log n}$ and
assume that the inequality $\EE[|D|] > U$ holds. Then,
the following inequalities hold:
\begin{align*}
& \EE[|D|] \left[  \log \frac{\EE[|D|]}{n}  - {m^{\ell}} - \log \left( \frac{n^2}{\EE[|D|]}  +  1 \right) -1\right] \\
& > U  \left[  \log \frac{10mn\log(m + 1)2^{m^{\ell}}}{\log n}  - m^{\ell} 
\right] 
  +  U \left[- \log \left( \frac{\log n}{10\log(m + 1)h}  +  1 \right) -1\right] \\
 & \geq U \left[  \log n - \log \left( \frac{\log n}{10\log(m + 1)2^{m^{\ell}}}  +  1 \right) \right] 
 -U  \left[  \log \log n\right]  \\
& > n^2 m \log(m + 1),
\end{align*}
which leads to a contradiction. This completes the proof of the lemma.
\end{proof}

We now have all the tools to prove Theorem~\ref{thm:main}.
Let $W(|D|)$ be the upper bound in Lemma~\ref{lem:LZA}. 
Since we have a probabilistic model and the fact that $W$ is concave in $|D|$,
the expected number of bits 
\[
\EE[l_{LDA}(A_n)] \leq W(\EE[|D|]).
\]\ignore{
\begin{align*}
& \EE[|D|] \left[ \log \left( (n + 1) \cdot \left( \frac{1}{W}  +  1 \right)  \cdot (\log (n + 1)  +  1)^2 \right) \right]\\
&  + \EE[|D|] \left[2   \log  \left(  \log \left( \frac{1}{W}  + 1 \right)  +  1  \right) 
  +  2  +   \lceil \log m \rceil   \right] \\
&  +   n \lceil \log nm \rceil,
\end{align*}
where $W = \frac{\EE[|D|]}{n^2}$.
}
Substituting the lower bound on $H(p)$ from Lemma~\ref{lem:lower} and rearranging terms, we have
\begin{align*}
 &  \max_{p(A_n)} \frac{\EE[l_{\LZA}(a_n)]  -H(p)}{n^2} \\
& =  \max_{p(A_n)} \frac{W(\EE[|D|])  -H(p)}{n^2} \\
\ignore{ & \leq \EE[|D|] \left[ \log \left( (n + 1) \cdot \left( \frac{n^2}{\EE[|D|]}  +  1 \right)  \cdot (\log (n + 1)  +  1)^2 \right) \right]\\
& +   \frac{\EE[|D|]}{n^2} \left[ 2   \log  \left(  \log \left( \frac{n^2}{\EE[|D|]}  + 1 \right)  +  1  \right)  +  2  +   \lceil \log m \rceil   \right] \\
   & -  \frac{\EE[|D|]}{n^2} \left[  \log \left( \frac{\EE[|D|]}{n} \cdot 
\left( \frac{n^2m}{\EE[|D|]}  +  1 \right) \right) - m^{\ell} -1\right]  +  \frac{\lceil \log nm \rceil }{n} \\}
   & \leq  \frac{\EE[|D|]}{n^2} \left[ \log \left( \frac{(n + 1)n}{\EE[|D|]} \cdot \left( \frac{n^2m}{\EE[|D|]}  +  1 \right)^2 \cdot (\log (n + 1)   +  1)^2 \right) \right] \\
& +   \frac{\EE[|D|]}{n^2} \left[2   \log  \left(  \log \left( \frac{n^2}{\EE[|D|]}  + 1 \right)  +  1  \right) +  m^{\ell}  +  4  +  \log m \right] 
  +  \frac{\lceil \log nm \rceil }{n}\\
   & = \cO \left(2^{m^{\ell}} \frac{\log \log n  +  m^{\ell}}{\log n}\right).
\end{align*}
The last equality follows from Lemma~\ref{lem:dictbound}.
As $n \to \infty$, the bound goes to $0$ and hence \LZA\ is a universal compression algorithm.

\end{document}